\documentclass[showpacs,twocolumn,english,10pt]{revtex4-1}

\usepackage{hyperref}

\usepackage{tikz} 
\usetikzlibrary{graphs} 
\usetikzlibrary{quotes} 
\usetikzlibrary{arrows} 
\usetikzlibrary{calc} 
\tikzset{
  snode/.style={circle,fill=red!20,draw,font=\sffamily\large\bfseries},
 onode/.style={circle,fill=blue!20,draw,font=\sffamily\large\bfseries},
  xnode/.style={circle,draw,font=\sffamily\large\bfseries}}

\usetikzlibrary{backgrounds,fit,decorations.pathreplacing,calc}
\tikzstyle{op} = [draw,fill=white,minimum size=1.5em]
\tikzstyle{ctr} = [draw,fill,shape=circle,minimum size=5pt,inner
sep=0pt]

\usepackage{amsthm}

\newtheorem{theorem}{Theorem}

\newtheorem{lemma}{Lemma}

\newtheorem{proposition}{Proposition}

\usepackage{physics}

\usepackage{subfig}

\makeatother

\usepackage{babel}

\newcommand{\nc}{\newcommand}

\def\i{\iota}

\def\r{\rho}

\def\ph{\varphi}

\def\ps{\psi}

 \nc{\bA}{{\bf A}} \nc{\bB}{{\bf B}} \nc{\bC}{{\bf C}}
 \nc{\bD}{{\bf D}} \nc{\bE}{{\bf E}} \nc{\bF}{{\bf F}}
 \nc{\bG}{{\bf G}} \nc{\bH}{{\bf H}} \nc{\bI}{{\bf I}}
 \nc{\bJ}{{\bf J}} \nc{\bK}{{\bf K}} \nc{\bL}{{\bf L}}
 \nc{\bM}{{\bf M}} \nc{\bN}{{\bf N}} \nc{\bO}{{\bf O}}
 \nc{\bP}{{\bf P}} \nc{\bQ}{{\bf Q}} \nc{\bR}{{\bf R}}
 \nc{\bS}{{\bf S}} \nc{\bT}{{\bf T}} \nc{\bU}{{\bf U}}
 \nc{\bV}{{\bf V}} \nc{\bW}{{\bf W}} \nc{\bX}{{\bf X}}
 \nc{\bZ}{{\bf Z}}

\nc{\cA}{{\cal A}} \nc{\cB}{{\cal B}} \nc{\cC}{{\cal C}}
\nc{\cD}{{\cal D}} \nc{\cE}{{\cal E}} \nc{\cF}{{\cal F}}
\nc{\cG}{{\cal G}} \nc{\cH}{{\cal H}} \nc{\cI}{{\cal I}}
\nc{\cJ}{{\cal J}} \nc{\cK}{{\cal K}} \nc{\cL}{{\cal L}}
\nc{\cM}{{\cal M}} \nc{\cN}{{\cal N}} \nc{\cO}{{\cal O}}
\nc{\cP}{{\cal P}} \nc{\cQ}{{\cal Q}} \nc{\cR}{{\cal R}}
\nc{\cS}{{\cal S}} \nc{\cT}{{\cal T}} \nc{\cU}{{\cal U}}
\nc{\cV}{{\cal V}} \nc{\cW}{{\cal W}} \nc{\cX}{{\cal X}}
\nc{\cZ}{{\cal Z}}

\def\bpp{\begin{proposition}}
\def\epp{\end{proposition}}
\def\bpf{\begin{proof}}
\def\epf{\end{proof}}
\def\bl{\begin{lemma}}
\def\el{\end{lemma}}
\def\bea{\begin{eqnarray}}
\def\eea{\end{eqnarray}}
\def\ox{\otimes}

\begin{document}

\title{Graph states of prime-power dimension from generalized CNOT quantum circuit}

\author{Lin Chen}
\affiliation{School of Mathematics and Systems Science, Beihang University, Beijing 100191, China}
\affiliation{International Research Institute for Multidisciplinary Science, Beihang University, Beijing 100191, China}

\author{D. L. Zhou}\email{zhoudl72@iphy.ac.cn} 
\affiliation{Beijing National Laboratory for Condensed Matter Physics,
and Institute of Physics, Chinese Academy of Sciences, Beijing 100190, China}

\date{\today}


\begin{abstract}
We construct multipartite graph states whose dimension is the power of a prime number. This is realized by the finite field, as well as the generalized controlled-NOT quantum circuit acting on two qudits. We propose the standard form of graph states up to local unitary transformations and particle permutations. The form greatly simplifies the  classification of graph states as we illustrate up to five qudits. We also show that some graph states are multipartite maximally entangled states in the sense that any bipartite of the system produces a bipartite maximally entangled state. We further prove that 4-partite maximally entangled states exist when the dimension is an odd number at least three or a multiple of four.
\end{abstract}

\maketitle

\section{Introduction}
\label{sec:introduction}

Maximal entanglement is the key ingredient in quantum teleportation, computing and the violation of Bell inequality. The maximally entangled state of two qubits can be created by
controlled-phase gate or controlled-not (CNOT) gate. In this sense, they have the same power to create entanglement. In fact, the two gates are related by local Hadamard
gates.  As we know, only one type of two-qubit unitary gates and single qubit
gates are enough to build a universal quantum circuit. A natural idea
is to use those gates to generate maximally entangled states in many qubit
case \cite{ffp08,gw10,vsk13}. The graph states and cluster states are
generated by applying two-qubit phase gates to an initially product state \cite{br01}.
Single-qubit gates are not involved in the generation. So the quantum circuit to create graph states is composed of pure control phase gates.
The graph states and continuous-variable cluster states are constructed to study one-way quantum computing \cite{br01,rb01,ZZXS2003,mvg06}.  They are useful for self-testing of nonlocal correlations \cite{mckague10} and their entanglement can be effectively evaluated by the Schmidt measure \cite{eb01}, relative entropy of entanglement and the geometric mesure of entanglement \cite{zch10,hm13}. Recently the graph states have been generalized to prime dimensions even in continuous variables, in terms of the encoding circuit and Hadamard matrices \cite{gkl15} and quantum codes and stabilizers \cite{cyz15}.

In this paper we study the multi-qudit graph state when the dimension $d=p^m$ is a power of a prime number $p$. It ensures the existence of finite field structure, and at the same time generalizes \cite{gkl15}. With the aid of the structure,
generalized CNOT gates are defined naturally. A general $N$ qudit state generated by a quantum circuit is constructed in Eq.~\eqref{eq:15}. To simplify this state, we propose a standard form of multiqubit state in Eq.~(\ref{eq:36}).
Our first main result is Theorem \ref{thm:main}, stating that the above two families are equivalent up to local unitary transformations and particle permutations. We also propose the dual graph state of the standard form in \eqref{eq:66}, and show that they are equivalent under local unitary transformation in Theorem \ref{thm:dual}. It further simplifies the structure of multiqudit graph states, and we classify them up to five parties.

Our main task is to
find out the maximally entangled state generated by the quantum circuit
composed by pure generalized CNOT gates. The task induces another relative problem: what states are called
maximally entangled states of many-qudit system? To avoid
confusion, let us constrain our discussions in many-qudit pure states.
The basic requirement for a many-qudit state being maximally entangled
satisfies that subsystem is entangled with the other, and
any single qudit is maximally entangled with the other. We
can further introduce that the many-qudit state
is a maximally entangled state  if any bipartite of systems produces a bipartite maximally entangled state \cite{ffp08,gw10,vsk13}. We will show that some graph states are multipartite maximally entangled states. We further prove that 4-partite maximally entangled states exist when the dimension is an odd number at least three or a multiple of four. This is another main result in our paper, as stated in Theorem \ref{thm:multiple4}. These results imply that the maximal entanglement is universal in high dimensions. We also construct a connection between maximal entanglement and the entropy problem recently proposed in \cite{chl14}. 

This paper is organized as follows. In Sec. \ref{sec:finite}, we will
introduce the generalized CNOT in the qudit case with the aid of the
structure of finite field, and then a quantum circuit composed pure
generalized CNOT gates is given. In Sec. \ref{sec:notimord}, we prove
that only bipartite graph states can be generalized from the quantum
circuit of pure generalized CNOT gates. In Sec. \ref{sec:entangl-prop-qudit}, we
analyze the maximal entanglement of these states. Finally, we 
give a summary of our results and open problems in Sec. \ref{sec:con}.

\section{Quantum circuit of pure generalized CNOT gates}
\label{sec:finite}

In this section we construct the generalized CNOT gates by two one-qudit operations $A(a_m)$ and $D(a_m)$ defined in Sec. \ref{sec:finite-field-gener}. They are mathematically realized by the known finite field and the commutation relations in Sec. \ref{sec:usef-relat-relat}. To illustrate the relations we construct Fig. \ref{fig:4}. Using the CNOT gates we construct the quantum circuit in Sec. \ref{sec:circuit}, and give an explicit example in Fig. \ref{fig:qc}. We will introduce a standard form of $N$-qudit graph state on finite field in \eqref{eq:36}, and show that any graph state is equivalent to the standard form up to local unitary transformations and particle permutations. To obtain a simpler classification of such states we propose Theorem \ref{thm:dual} and demonstrate it by states up to five systems respectively in Fig. \ref{fig:6} to \ref{fig:9}.

\subsection{finite field and generalized CNOT gates}
\label{sec:finite-field-gener}

As is well known, when $d$ is the power of a prime number, i.e.,
\begin{equation}
  d=p^{n},
\end{equation}
where $p$ is prime, and $n$ is a positive integer, there is a field
$F_{d}$. Note that the field $F_{d}$ is unique up to isomorphism. The
elements of the Field $F_{d}$ are denoted as
$\left\{ a_{i},i\in\left\{ 0,1,\ldots,d-1\right\} \right\} $, where
$a_{0}\equiv0$ and $a_{1}\equiv1$ are the units for the sum and the
product operations respectively.

We introduce a $d$-dimensional Hilbert space $H_{d}$ with a natural
orthonormal basis $\left\{ | a_{i}\rangle\right\}$. With the aid of
the sum and product operations in the field, two classes of basic
one-qudit operations are defined
\begin{eqnarray}
  A\left(a_{m}\right)| a_{i}\rangle & = & | a_{i}+a_{m}\rangle, \label{eq:18}\\
  D\left(a_{m}\right)| a_{i}\rangle & = & | a_{m}a_{i}\rangle.\label{eq:19}
\end{eqnarray}
Obviously, the operation $A\left(a_{m}\right)$ is unitary for any $m$.
If $a_{m}\neq{0}$, then $D\left(a_{m}\right)$ is also unitary.

Since $F_{d}$ is an Abelian group under the operation $+$, then we
have
\begin{equation}
  A\left(a_{m}\right)| s\rangle=| s\rangle,
\end{equation}
where
\begin{equation}
  | s\rangle=\frac{1}{\sqrt{d}}\sum_{i}| a_{i}\rangle.
\end{equation}

We introduce the generalized CNOT gate from qudit $m$ to qudit
$n$ labeled by $a_{k}$ defined by
\begin{equation}
  C_{m n}\left(a_{k}\right)| a_{i}\rangle_{m}|
  a_{j}\rangle_{n} = | a_{i}\rangle_{m}| a_{j}+a_{i}a_{k}\rangle_{n}, \label{eq:2}
\end{equation}
where qudit $m$ is the control qudit, and qudit $n$ is the target
qudit.

First, we notice that
\begin{equation}
  \label{eq:38}
  C_{mn}(a_{k}) \ket{s,a_{j}}_{mn} = A_{n}(a_{j}) D_{n}(a_{k}) \ket{B}_{mn},
\end{equation}
where
\begin{equation}
  \label{eq:39}
   \ket{B}_{mn} = \frac{1}{\sqrt{d}} \sum_{i}
  \ket{a_{i},a_{i}}_{mn}.
\end{equation}

In addition, when $d=2$ and $a_{k}=1$, the gate $C_{mn}(1)$ is the
CNOT gate. Therefore any $C_{mn}(a_{k})$ with $a_{k}\neq{0}$ is
a generalized CNOT gate, which can generate the two-qudit maximal
entangled state from a separable state.

\subsection{Commutation relations for related unitary transformations}
\label{sec:usef-relat-relat}

Before simplifying the above quantum circuit and investigating the
properties of the generated states, let us first calculate the basic
commutation relations for related unitary transformations widely used throughout the paper. The proof of these relations will be given in Appendix \ref{sec:proof-comm-relat}.

\subsubsection{One qudit case}
\label{sec:one-qudit-unitary}

According to the definitions given in Eq.~(\ref{eq:18}) and
Eq.~(\ref{eq:19}), we have
\begin{eqnarray}
  \label{eq:20}
  A_{m}(a_{i})A_{m}(a_{j}) & = & A_{m}(a_{i}+a_{j}), \\
  D_{m}(a_{i}) D_{m}(a_{j}) & = & D_{m}(a_{i}a_{j}). \label{eq:21}
\end{eqnarray}

The commutation relations between $A_{m}$ and $D_{m}$ are
\begin{equation}
  \label{eq:22}
  D_{m}(a_{i}) A_{m}(a_{j}) = A_{m}(a_{i} a_{j}) D_{m}(a_{i}).
\end{equation}

In addition, we also have
\begin{equation}
  \label{eq:34}
  A_{m}(0) = D_{m}(1) = I_{m}.
\end{equation}

\subsubsection{Two-qudit case}
\label{sec:two-qudit-relations}

The first set of relations are
\begin{eqnarray}
  \label{eq:24}
  C_{mn}(a_{i}) A_{m}(a_{j}) & = & A_{n}(a_{i} a_{j}) A_{m}(a_{j})
                                   C_{mn}(a_{i}), \\
  C_{mn}(a_{i}) A_{n}(a_{j}) & = & A_{n}(a_{j})
                                   C_{mn}(a_{i}), \label{eq:25} \\
  C_{mn}(a_{i}) D_{m}(a_{j}) & = & D_{m}(a_{j})
                                   C_{mn}(a_{j} a_{i}), \label{eq:26} \\
  C_{mn}(a_{i}) D_{n}(a_{j}) & = & D_{n}(a_{j})
                                   C_{mn}( a_{j}^{-1} a_{i}). \label{eq:27}
\end{eqnarray}

The second set of relations includes two equations. The first equation
is
\begin{equation}
  \label{eq:28}
  C_{mn}(a_{i}) C_{mn}(a_{j}) = C_{mn}(a_{i} + a_{j}),
\end{equation}
which is easy to prove but important in simplifying our graph sates.

The second equation is
\begin{eqnarray}
  & & C_{mn}(a_{i}) C_{nm}(a_{j}) \nonumber\\
  & = &
        \begin{cases}
          D_{m}(A^{-1}) D_{n}(A) C_{nm}(A a_{j}) C_{mn}(A^{-1} a_{i}) &
          \text{if } A\neq 0,\\
          W_{mn} D_{m}(a_{i}) D_{n}(a_{j}) C_{mn}(a_{j}^{-1}), & \text{if } A=0,
        \end{cases} \nonumber\\
  \label{eq:30}
\end{eqnarray}
where $A=1+a_{i}a_{j}$, $W_{mn}$ is the swap gate between the qudits
$m$ and $n$.

\subsubsection{Three-qudit case}
\label{sec:three-qudit-relat}

The relations for three qudits are given by
\begin{eqnarray}
  \label{eq:31}
  C_{mn}(a_{i}) C_{ml}(a_{j}) & = & C_{ml}(a_{j}) C_{mn}(a_{i}), \\
  C_{mn}(a_{i}) C_{l n}(a_{j}) & = & C_{l n}(a_{j})
                                    C_{mn}(a_{i}), \label{eq:32} \\
  C_{nl}(a_{j}) C_{mn}(a_{i}) & = & C_{ml}(a_{i} a_{j}) C_{mn}(a_{i})
                                    C_{nl}(a_{j}). \label{eq:33}
\end{eqnarray}

Here we use two circuits to represent Eq.~(\ref{eq:33}) as shown in
Fig.~\ref{fig:4}.
\begin{figure}[htbp]
  \centering
    \begin{tikzpicture}[thick]
    \matrix[row sep=0.4cm, column sep=0.6cm] (circuit) {
      \node (q1) {}; &[-0.3cm] \node[ctr] (C12) {}; &  &[-0.3cm]
      \coordinate (end1); \\
      \node (q2) {}; & \node[op] (A12) {$a_{i}$}; &  \node[ctr] (C23) {}; &
      \coordinate (end2);\\
      \node (q3) {}; &  &  \node[op] (A23)
      {$a_{j}$}; &  \coordinate (end3); \\
    };
    \begin{pgfonlayer}{background}
      \draw[thick] (q1) -- (end1) (q2) -- (end2) (q3) -- (end3) (C12)
      -- (A12) (C23) -- (A23);
    \end{pgfonlayer}
  \draw [ultra thick,blue,<->] (1.8,0) -- (2.6,0);
  \end{tikzpicture}
    \begin{tikzpicture}[thick]
    \matrix[row sep=0.4cm, column sep=0.6cm] (circuit) {
      \node (q1) {}; &[-0.3cm] & \node[ctr] (C12) {}; & \node[ctr]
      (C13) {};  &[-0.3cm]
      \coordinate (end1); \\
      \node (q2) {}; & \node[ctr] (C23) {}; & \node[op] (A12) {$a_{i}$}; &  &
      \coordinate (end2);\\
      \node (q3) {}; &   \node[op] (A23)
      {$a_{j}$}; & & \node[op] (A13) {$a_{i}a_{j}$}; &  \coordinate (end3); \\
    };
    \begin{pgfonlayer}{background}
      \draw[thick] (q1) -- (end1) (q2) -- (end2) (q3) -- (end3) (C12)
      -- (A12) (C23) -- (A23) (C13) -- (A13);
    \end{pgfonlayer}
  \end{tikzpicture}
  \caption{Circuit representation of Eq.~(\ref{eq:33})\label{fig:4}}
\end{figure}
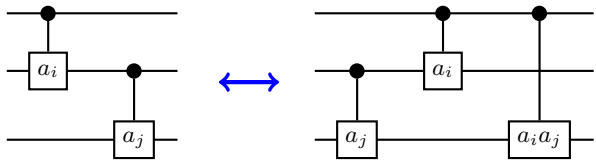

\subsection{Quantum circuit based on controlled gates}
\label{sec:circuit}

Since a controlled gate can generate a two-qudit maximally entangled
state, and a two-qudit gate is enough to entangle a complex quantum
circuit, a natural generalization is to apply the controlled gates to
generate many-qudit maximally entangled state in a quantum circuit.

A quantum circuit based on the controlled gates $C_{mn}(a_{k})$ is an
$N$-qudit circuit with a series of controlled gates operating on, see an
example as shown in Fig.~\ref{fig:qc}. A general $N$ qudit ($d=p^{m}$)
state generated by a quantum circuit is
\begin{eqnarray}
  \label{eq:15}
  \vert G\rangle & = & \otimes_{\tau} C_{m_{\tau} n_{\tau}}(b_{\tau})
                       \otimes_{i} \vert c_{i}\rangle_{i},
\end{eqnarray}
where $c_{i}\in\{s,0\}$, $b_{\tau}\in F_{d}$,
$\tau\in\{1,2,\ldots,M\}$ with $M$ being the number of the
controlled gates, and $(m_{\tau},n_{\tau})\in\{1,2,\ldots,N\}$.

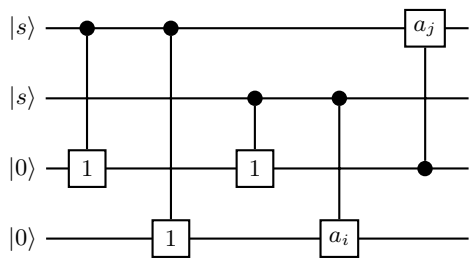
\begin{figure}[htbp]
  \centering
  \begin{tikzpicture}[thick]
    \matrix[row sep=0.4cm, column sep=0.6cm] (circuit) {
      \node (q1) {$\ket{s}$}; &[-0.3cm] \node[ctr] (C13) {}; &
      \node[ctr] (C14) {}; & & & \node[op] (A31) {$a_{j}$}; &[-0.3cm]
      \coordinate (end1); \\
      \node (q2) {$\ket{s}$}; & & & \node[ctr] (C23) {}; & \node[ctr]
      (C24) {}; & &
      \coordinate (end2);\\
      \node (q3) {$\ket{0}$}; & \node[op] (A13) {1}; & & \node[op] (A23)
      {1}; & & \node[ctr] (C31) {}; &
      \coordinate (end3); \\
      \node (q4) {$\ket{0}$}; & & \node[op] (A14) {1}; & & \node[op]
      (A24) {$a_{i}$}; & &
      \coordinate (end4); \\
    };
    \begin{pgfonlayer}{background}
      \draw[thick] (q1) -- (end1) (q2) -- (end2) (q3) -- (end3) (q4)
      -- (end4) (C13) -- (A13) (C14) -- (A14) (C23) -- (A23) (C24) --
      (A24) (C31) -- (A31);
    \end{pgfonlayer}
  \end{tikzpicture}
  \caption{A quantum circuit to generate 4-qudit graph
    state.\label{fig:qc}}
\end{figure}

A central problem is to investigate the possible types of entangled
states through a series of the above controlled operations with some
given initial states.

The difficulties in simplifying the circuit lies in the facts that the
number of controlled gates $M$ may be very large, and these controlled gates
do not commute with each other in general.

\section{Graph state on finite field}

\label{sec:notimord}

According to the initial state of an $N$-qudit circuit state in
Eq.~(\ref{eq:15}), we divide the $N$ qudits into two sets: the set of
qudits with the initial state $\ket{s}$ and the set of qudits with the
initial state $\ket{0}$, denoted as $S$ and $O$ respectively.

Now we introduce a standard form of $N$-qudit graph state on finite field as
\begin{equation}
  \label{eq:36}
  \prod_{i\in S, j\in O}  C_{ij}(b_{ij}) \ket{S} ,
\end{equation}
where $b_{ij}\in{F_{d}}$, and
\begin{equation}
  \label{eq:40}
  \ket{S} = \otimes_{i\in S} {\vert
    s\rangle}_{i} \otimes_{j\in O} {\vert 0\rangle}_{j}.
\end{equation}
This state is called a graph state because the time ordering of the
controlled gates is unrelated, and it is can be represented as a
directional bipartite graph. An example of a graph state for $N=7$ and
the set $S=\{1,2,3\}$ is demonstrated in Fig.~\ref{fig:10}.

\begin{figure}[htbp]
  \centering
  \begin{tikzpicture}[->,>=stealth',shorten >=1pt,auto,node
    distance=2cm, thick]
    \node[snode] (1) {1}; %
    \node[snode] (2) [right of=1] {2}; %
    \node[snode] (3) [right of=2] {3}; %
    \node[onode] (4) [left=1cm,below of=1] {4}; %
    \node[onode] (5) [right of=4] {5}; %
    \node[onode] (6) [right of=5] {6}; %
    \node[onode] (7) [right of=6] {7};
    \path[every node/.style={font=\sffamily\small}] %
    (1) edge node [above] {} (4) edge node [above] {}
    (5) edge node [above] {} (6) edge node [above] {}
    (7)                   %
    (2) edge node [above] {} (4) edge node [above] {}
    (5) edge node [above] {} (6) edge node [above] {}
    (7)                %
    (3) edge node [above] {} (4) edge node [above] {}
    (5) edge node [above] {} (6) edge node [above] {}
    (7);
  \end{tikzpicture}
  \caption{A bipartite graph state with $N=7$ and $S=\{1,2,3\}$, and
    the labels $\{b_{ij}\}$ are omitted.\label{fig:10}}
\end{figure}
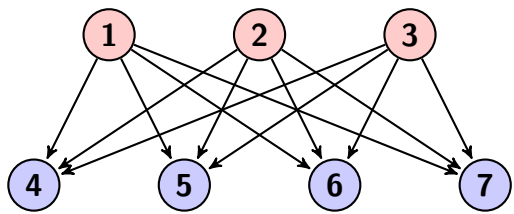

One of our central results is the following theorem:
\begin{theorem}
\label{thm:main}
  Any state in Eq.~(\ref{eq:15}) is equivalent to the standard form in
  Eq.~(\ref{eq:36}) up to local unitary transformations and particle
  permutations.
\end{theorem}

A direct way to prove the above theorem is to show a state in the
standard form under the action of any generalized CNOT gate will still
be a standard one. More precisely, we only need to show
\begin{eqnarray}
  \label{eq:37}
  &  & C_{mn}(a_{r}) \prod_{i\in S,j\in O}   C_{ij}(b_{ij})  \ket{S} \nonumber\\
  & = & W \prod_{k=1}^{N} D(c_{k}) \prod_{i\in S,j\in O}   C_{ij}(d_{ij})  \ket{S},
\end{eqnarray}
where $m,n\in\{1,2,\ldots,N\}$ and
$a_{r},b_{ij},c_{k},d_{ij}\in F_{d}$. In fact, this can be proved by
directly applying the commutation relations given in the last section.

Here we give another more concise proof.

\begin{proof}

  Let the initial state $\ket{S}$ be
  $\ket{s}^{\otimes k} \ket{0}^{\otimes (N-k)}$ with $k\in [1,N-1]$. Any graph
  state can be expressed as
\begin{eqnarray}
\label{eq:bi}
  &&
     \ket{G}
     \notag\\
  &=&
      \bigg( \otimes^M_{\alpha=1} C_{m_{\alpha}n_{\alpha}}(b_{\alpha}) \bigg)
      \notag\\
  &&
     \bigg(
     d^{-k/2}
     \sum_{j_1,\cdots,j_k} \ket{a_{j_1},\cdots,a_{j_k}} \ket{0}^{\otimes (N-k)} \bigg)
     \notag\\
  &=&
      d^{-k/2}
      \sum_{j_1,\cdots,j_k}
      \bigg|\sum^k_{i=1} c_{i,1} a_{j_i} \bigg\rangle
      \otimes
      \cdots
      \otimes
      \bigg|\sum^k_{i=1} c_{i,N} a_{j_i} \bigg\rangle,
      \notag\\
\end{eqnarray}
where $c_{i,q}\in\mathbf{F}_d$, and the $k\times N$ matrix $[c_{i,q}]$
has rank $k$. Up to the permutation of vertices, we may assume that
the first $k$ column vectors in $[c_{i,q}]$ are linearly independent.
So there are elements $b_{1,1},\cdots,b_{k,N}\in\mathbf{F}_d$ such
that
\begin{eqnarray}
  \ket{G}&=&\otimes^k_{j=1} \otimes^N_{l=k+1} C_{j,l} (b_{j,l})
             \notag\\
         &&
            d^{-k/2}
            \sum_{j_1,\cdots,j_k}
            \bigg|\sum^k_{i=1} c_{i,1} a_{j_i} \bigg\rangle_1
            \otimes\cdots
            \notag\\
         &\otimes&
               \bigg|\sum^k_{i=1} c_{i,k} a_{j_i} \bigg\rangle_k
               \otimes
               \ket{0}_{k+1} \otimes \cdots \otimes \ket{0}_N
               \notag\\
         &=&(\otimes^k_{j=1} \otimes^N_{l=k+1}) C_{j,l} (b_{j,l})
             \ket{s}^{\otimes k} \ket{0}^{\otimes (N-k)}
             .
\end{eqnarray}
Hence we can generate $\ket{G}$ by performing the bidirectional gate
$(\otimes^k_{j=1} \otimes^N_{l=k+1}) C_{j,l} (b_{j,l})$ on the initial
state $\ket{s}^{\otimes k} \ket{0}^{\otimes (N-k)}$. The time order of
the gates $C_{j,l} (b_{j,l}$ is random, because they commute. This
completes the proof.
  
\end{proof}

The main conclusion from the above theorem is that up to local unitary
transformations and particle permutations all the states generated by
the controlled gate circuit are the directional bipartite graph states,
and the graph contains only the edges from $\ket{s}$ to $\ket{0}$,
which greatly simplifies our investigations on possible types of
entanglement created by the controlled gate circuit.

The dual graph state for the graph state specified by Eq.~(\ref{eq:36})
and Eq.~(\ref{eq:40}) is
\begin{equation}
  \label{eq:66}
  \prod_{i\in O,j\in S} C_{ij}(b_{ji}) \ket{O},
\end{equation}
where
\begin{equation}
  \label{eq:67}
  \ket{O} = \otimes_{i\in O} \ket{s}_{i} \otimes_{j\in S} \ket{0}_{j}.
\end{equation}

\begin{theorem}
\label{thm:dual}
  The two graph states given in Eq.~(\ref{eq:36}) and
  Eq.~(\ref{eq:66}) for two dual graphs are local unitary equivalent.
\end{theorem}

\begin{proof}
  For a finite field with $d=p^{n}$ and $p$ a prime, the element is
  denoted as $a=\sum_{i=0}^{n-1} a_{i} \alpha^{i}$, where
  $a_{i}\in\{0,1,\ldots,p-1\}$ and $\alpha$ is one root of some
  irreducible polynomial equation. The element in the finite field can
  be denoted as a vector $\vec{a}$. Then we introduce the discrete
Fourier transformation of the states $\{\ket{a}\}$ as
\begin{equation}
  \label{eq:48}
  {\ket{\vec{b}}}^{\prime} = \frac{1}{\sqrt{d}} \sum_{\vec{a}}
  \omega^{\vec{a}\vdot \vec{b}} \ket{\vec{a}},
\end{equation}
where
\begin{equation}
  \label{eq:64}
  \vec{a}\vdot \vec{b} = \sum_{i=0}^{n-1} a_{i} b_{i} \mod p.
\end{equation}
Therefore we define the Hadmard transformation as
\begin{equation}
  \label{eq:49}
  H  =   \sum_{\vec{b}} {\ket{\vec{b}}}^{\prime} \bra{\vec{b}}
     =  \frac{1}{\sqrt{d}} \sum_{\vec{a},\vec{b}} \omega^{\vec{a}\vdot
          \vec{b}} \op{\vec{a}}{\vec{b}}.
\end{equation}

Then
\begin{equation}
  \label{eq:50}
  C_{mn}(\vec{d}) = \sum_{\vec{a},\vec{b}} \op{\vec{a}} \otimes
  \op{\vec{b}+ \overrightarrow{da}}{\vec{b}}.
\end{equation}
Therefore
\begin{equation}
  \label{eq:51}
  H_{n} C_{mn}(\vec{d}) H_{n}^{\dagger} = \sum_{\vec{a},\vec{b}}
  \omega^{\overrightarrow{da} \vdot \vec{b}}
  \op{\vec{a},\vec{b}}.
\end{equation}
Note that
\begin{eqnarray*}
  \label{eq:59}
  \overrightarrow{da} \vdot \vec{b} & = & d_{i} a_{j} b_{k}
                                          \alpha^{i+j}(\alpha^{k}) \\
                                    & = & d_{i} a_{j} b_{k} \alpha^{i+j-k}(1) \\
                                    & = & d_{i} a_{j} b_{k} \alpha^{i
                                          + (n-1 -k) - (n-1-j)}(1) \\
                                    & = & d_{i} a_{j} b_{k} \alpha^{i
                                          + (n-1-k)}(\alpha^{n-1-j}) \\
                                    & = & \overrightarrow{db^{\prime}}
                                          \vdot \vec{a}^{\prime},
\end{eqnarray*}
where
\begin{eqnarray}
  \label{eq:60}
  b^{\prime}_{n-1-k} & = & b_{k}, \\
  a^{\prime}_{n-1-j} & = & a_{j}.
\end{eqnarray}
So we introduce local unitary transformation
\begin{equation}
  \label{eq:61}
  V \ket{\vec{a}} = \ket{\vec{a}^{\prime}}.
\end{equation}
Therefore we have
\begin{eqnarray}
  \label{eq:62}
  &  & H_{m}^{\dagger} V_{m} V_{n} H_{n} C_{mn}(\vec{d}) H_{n}^{\dagger}
       V_{n}^{\dagger} V_{m}^{\dagger} H_{m} \nonumber\\
  & = & H_{m}^{\dagger} V_{m} V_{n} \sum_{\vec{a},\vec{b}}
        \omega^{\overrightarrow{db^{\prime}} \vdot \vec{a}^{\prime}}
        \op{\vec{a},\vec{b}}
        V_{n}^{\dagger} V_{m}^{\dagger} H_{m}\nonumber\\
  & = & H_{m}^{\dagger} \sum_{\vec{a},\vec{b}}
        \omega^{\overrightarrow{db^{\prime}} \vdot \vec{a}^{\prime}}
        \op{\vec{a}^{\prime},\vec{b}^{\prime}} H_{m} \nonumber\\
  & = & H_{m}^{\dagger} \sum_{\vec{a},\vec{b}}
        \omega^{\overrightarrow{db} \vdot \vec{a}}
        \op{\vec{a},\vec{b}} H_{m} \nonumber \\
  & = &  \sum_{\vec{a},\vec{b}}
        \op{\vec{a} + \overrightarrow{db},\vec{b}}{\vec{a},\vec{b}}
        \nonumber \\
  & = & C_{nm}(\vec{d}).
\end{eqnarray}
In addition,
\begin{eqnarray}
  \label{eq:63}
  V_{n} H_{n} \ket{0}_{n} & = & \ket{s}_{n}, \\
  H_{m}^{\dagger} V_{m} \ket{s}_{m} & = & \ket{0}_{m}.
\end{eqnarray}
Therefore we have
\begin{eqnarray}
  \label{eq:68}
  &  & \otimes_{m\in S} H_{m}^{\dagger} V_{m} \otimes_{n\in O} V_{n} H_{n}
  \prod_{i\in S,j\in O}
  C_{ij}(b_{ij}) \ket{S} \nonumber\\
& = &  \prod_{i\in O,j\in S} C_{ij}(b_{ji}) \ket{O}.
\end{eqnarray}
This completes our proof.

\end{proof}

This theorem implies that we can restrict ourselves in the cases where
the cardinality of $S$ is less than the cardinality of $O$, i.e.
$[N/2]$.

\subsection{Examples}
\label{sec:examples}

Now let us apply the above theorems to study the possible types of
entanglement generated by the controlled gates for $N=3,4,5$ with the
help of Eq.~(\ref{eq:26}) and Eq.~(\ref{eq:27}).

There is only one type of two qudit graph state, which is the qudit
Bell state:
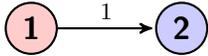
\begin{figure}[htbp]
  \centering
  \begin{tikzpicture}[->,>=stealth',shorten >=1pt,auto,node
    distance=2cm, thick]
    \node[snode] (1) {1}; %
    \node[onode] (2) [right of=1] {2}; %
    %
    \path[every node/.style={font=\sffamily\small}] (1) edge node
    [above] {$1$} (2);
  \end{tikzpicture}
  \caption{Two qudit graph.\label{fig:6}}
\end{figure}

There is also one type of three qudit graph state, which is a
generalized GHZ state:
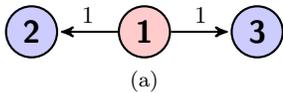
\begin{figure}[htbp]
  \centering %
  \subfloat[][] {
    \begin{tikzpicture}[->,>=stealth',shorten >=1pt,auto,node
      distance=1.5cm, thick] \node[snode] (1) {1}; %
      \node[onode] (2) [left of=1] {2}; %
      \node[onode] (3) [right of=1] {3}; %
      %
      \path[every node/.style={font=\sffamily\small}] %
      (1) edge node [above] {$1$} (2) %
      (1) edge node [above] {$1$} (3);
    \end{tikzpicture}
  }  
  \caption{Three qudit graph.\label{fig:7}}
\end{figure}

There are two types of four qudit graph states. One is four qudit GHZ
state in Fig.~\ref{fig:8}~(a). The other type in Fig.~\ref{fig:8}~(b) has a more fruitful configuration, which will be studied in next section.
\begin{figure}[htbp]
  \centering %
  \subfloat[][] {
    \begin{tikzpicture}[->,>=stealth',shorten >=1pt,auto,node
      distance=1.5cm, thick]                %
      \node[snode] (1) {1}; %
      \node[onode] (2) [right of=1] {2}; %
      \node[onode] (3) [below of=2] {3}; %
      \node[onode] (4) [left of=3] {4};
      \path[every node/.style={font=\sffamily\small}] %
      (1) edge node [above] {$1$} (2) %
      (1) edge node [above] {$1$} (3) %
      (1) edge node [right] {$1$} (4);
    \end{tikzpicture}
  }  \quad
\subfloat[][] {
    \begin{tikzpicture}[->,>=stealth',shorten >=1pt,auto,node
      distance=1.5cm, thick]                %
      \node[snode] (1) {1}; %
      \node[onode] (2) [right of=1] {2}; %
      \node[snode] (3) [below of=2] {3}; %
      \node[onode] (4) [left of=3] {4};
      \path[every node/.style={font=\sffamily\small}] %
      (1) edge node [above] {$1$} (2) %
      (1) edge node [right] {$1$} (4) %
      (3) edge node [above] {$1$} (4) %
      (3) edge node [right] {$a_{r}$} (2);
    \end{tikzpicture}
  }
  \caption{Four qudit graph states.\label{fig:8}}
\end{figure}
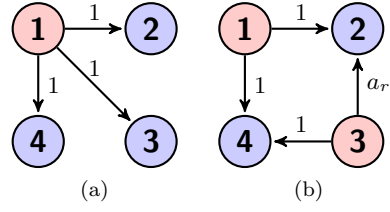

There are also two types of five qudit graph states:
\begin{figure}[htbp]
  \centering %
  \subfloat[][] {
    \begin{tikzpicture}[->,>=stealth',shorten >=1pt,auto,node
      distance=1.5cm, thick]                %
      \node[snode] (1) {1}; %
      \node[onode] (2) [right of=1] {2}; %
      \node[onode] (3) [above of=1] {3}; %
      \node[onode] (4) [left of=1] {4};  %
      \node[onode] (5) [below of=1] {5};  %
      \path[every node/.style={font=\sffamily\small}] %
      (1) edge node [above] {$1$} (2) %
      (1) edge node [left] {$1$} (3) %
      (1) edge node [below] {$1$} (4) %
      (1) edge node [right] {$1$} (5); %
    \end{tikzpicture}
  }  
\quad
  \subfloat[][] {
    \begin{tikzpicture}[->,>=stealth',shorten >=1pt,auto,node
      distance=1.5cm, thick]                %
      \node[onode] (1) {1}; %
      \node[snode] (2) [right of=1] {2}; %
      \node[onode] (3) [above of=1] {3}; %
      \node[snode] (4) [left of=1] {4};  %
      \node[onode] (5) [below of=1] {5};  %
      \path[every node/.style={font=\sffamily\small}] %
      (2) edge node [above] {$1$} (1) %
      (2) edge node [above] {$a_{r}$} (3) %
      (2) edge node [above] {$1$} (5) %
      (4) edge node [above] {$1$} (1) %
      (4) edge node [above] {$1$} (3) %
      (4) edge node [above] {$a_{j}$} (5); %
    \end{tikzpicture}
  }  
  \caption{Five qudit graph states.\label{fig:9}}
\end{figure}
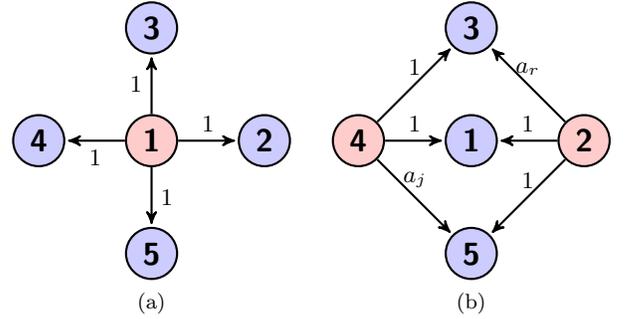

\section{Entanglement properties of qudit graph states}
\label{sec:entangl-prop-qudit}

In this section we study the maximal entanglement of graph states defined in previous sections.
The state in Fig.~\ref{fig:8}~(b) can be written as
\bea
\label{eq:psgamma}
\ket{\ps(a_r)}=d^{-1} \sum_{i,k} \ket{a_i,a_i+a_r a_k,a_k,a_i+a_k},
\eea
where $a_r\in F_d$, the dimension $d=p^n$ with a prime $p$ and positive integer $n$. We have
\bl
\label{le:pha}
$\ket{\ps(a_r)}$ is a maximally entangled state when $a_r\in F_d\setminus\{a_0,a_1\}$.
\el
\bpf
Since $F_d$ is a field and $a_r\in F_d\setminus\{a_0,a_1\}$, we have $F_d=a_r F_d=a_i+F_d$ for any $a_i\in F_d$. One can easily verify that all six bipartite reduced density operators of $\ket{\ps(a_r)}$ are maximally mixed states. 
So $\ket{\ps(a_r)}$ is a maximally entangled state.
\epf
If $n=1$ then $d$ is a prime number. This case has been studied in \cite{gkl15} and is a special case of the lemma. The case $d=2$ is excluded in the lemma, and it coincides with the known result that 4-qubit maximally entangled state does not exist \cite{gw10}.
we demonstrate them by a simple example. We set $a_r=2$, $a_j=j$ and $d=4$ in \eqref{eq:psgamma},  and obtain
\bea
\ket{\ps(2)} 
&=&{1\over4}(\ket{0000}+\ket{0211}+\ket{0322}+\ket{0133}
\notag\\
&+&\ket{1101}+\ket{1310}+\ket{1223}+\ket{1032}
\notag\\
&+&\ket{2202}+\ket{2013}+\ket{2120}+\ket{2331}
\notag\\
&+&\ket{3303}+\ket{3112}+\ket{3021}+\ket{3230}),
\eea
by using the computation rule in Table \ref{tab:field4}. On the other hand, Lemma \ref{le:pha} does not hold when $d$ is replaced  by any integer which is not a prime power.  

\begin{table}
  \caption{
  \label{tab:field4}
The two tables respectively account for the addition and multiplication operations for $F_4$. The proof of Lemma \ref{le:pha} also holds when $F_d$ is replaced by any finite domain, because it coincides with the finite field \cite{domain}. }
\begin{tabular}{|c|c|c|c|c|}
   \hline
   $+$   & 0    & 1 & 2  & 3        \\\hline
   0       & 0    & 1 & 2  & 3        \\\hline
   1       & 1    & 0 & 3  & 2        \\\hline
   2       & 2    & 3 & 0  & 1        \\\hline
   3       & 3    & 2 & 1  & 0        \\\hline
 \end{tabular}
 
 \begin{tabular}{|c|c|c|c|c|}
   \hline
   $\times$   & 0    & 1 & 2  & 3        \\\hline
   0       & 0    & 0 & 0  & 0        \\\hline
   1       & 0    & 1 & 2  & 3        \\\hline
   2       & 0    & 2 & 3  & 1        \\\hline
   3       & 0    & 3 & 1  & 2        \\\hline
 \end{tabular}
 \end{table}

Next we give an example of maximal entanglement beyond the primer-power dimension. The state
\bea
\ket{P'}
=
d^{-1} \sum_{i,k} \ket{i,i-k,k,i+k}
\eea
appeared in \cite{gkl15}, in which $d$ was considered as a prime number.
We point out that the state can be defined for any integer $d$. One can straightforwardly show that $\ket{P'}$ is a maximally entangled state for any odd $d>2$, and is not a maximally entangled state for any even $d>1$.
The two families of states $\ket{\ps(a_r)}$ and $\ket{P'}$ show that 4-partite maximally entangled states are universal in high dimensional spaces. Indeed we have

\begin{theorem}
\label{thm:multiple4}
The maximally entangled 4-partite pure state exists when the dimension $d$ is an odd number at least three, or a multiple of four.
\end{theorem}
\bpf
The state $\ket{P'}$ validates the assertion when $d$ is an odd number at least three. It remains to prove the assertion when $d$ is a multiple of four. We may assume $d=2^m \Pi^k_{j=1} p_j$ where $m\ge2$, $k\ge0$, and $p_j\ge3$ are prime numbers. 
The first assertion implies that we may assume $\ket{\ps_j}$ as the maximally entangled state with every system of dimension $p_j$. From Lemma \ref{le:pha}, we may assume $\ket{\ph}$ as the maximally entangled state with every system of dimension $2^m$. We define a new 4-partite pure state as the tensor product of these states, i.e.,
$
\ket{\ph} \ox \ket{\ps_1} \ox \cdots \ox \ket{\ps_k}.
$
One can straightforwardly show that this is a maximally entangled state.
\epf

The above proof indeed shows an analytical way of constructing the 4-partite maximally entangled states in designated dimensions. In spite of the above results,
we do not have any example of 4-partite maximally entangled state with dimension equal to the multiple of two and any positive odd number. We conjecture they might not exist. This is true when the odd number is one \cite{gw10}. So the first challenge is to construct a 4-partite maximally entangled state with dimension $6$. It easily reminds us of the construction of mutually unbiased basis of dimension $6$, which is a long-standing problem in quantum physics.

Finally as a more independent interest, we construct the connection between maximal entanglement and the entropy problem recently proposed in \cite{chl14}. The problem asks to construct (or exlcude the existence of) a tripartite quantum state $\r_{ABC}$ such that $\rank\r_{AB}>\rank\r_{AC}\cdot\rank\r_{BC}$. The problem turns out to be hard and constructing the connection might be helpful to finding out its solution.

\bl
\label{le:2mm=mm}
Let $\r_{ABC}$ be a tripartite state whose bipartite reduced density matrices are all maximally mixed states ${1\over d^2}I_d\ox I_d$. Then
\\
(i) $\r_{ABC}$ exists and $\rank\r_{ABC}\ge d$.
\\
(ii) 
The maximally entangled 4-partite pure state exists if and only if there is a $\r_{ABC}$ such that $\rank\r_{ABC}=d$.
\el
\bpf
(i) A trivial example is $\r_{ABC}={1\over d^3}I_d\ox I_d\ox I_d$. Let $\ket{\ps}_{ABCD}$ be the purification of $\r_{ABC}$. Then $\rank\r_{AB}=\rank\r_{CD}=d^2 \le \rank\r_C \rank\r_D$. Since $\rank\r_C=d$, we have $\rank\r_{ABC}=\rank \r_D\ge d$.

(ii) We prove the ``if'' part. Suppose there is a tripartite state $\r_{ABC}$ of rank $d$, whose bipartite reduced density matrices are all maximally mixed states ${1\over d^2}I_d\ox I_d$. Let $\ket{\ps}_{ABCD}$ be the purification of $\r_{ABC}$. So $\ket{\ps}_{ABCD}\in\cH$ is maximally entangled. The ``only if'' part can be similarly proved.
This completes the proof.
\epf

\section{conclusions}
\label{sec:con}

We have constructed multipartite graph states with prime-power dimesnion using the generalized CNOT quantum circuit. We have proven that the graphs states are equivalent to a simple and operational standard form up to local unitary transformations and particle permutations. We also showed that some graph states are multipartite maximally entangled states, and that 4-partite maximally entangled states exist when the dimension is an odd number at least three or a multiple of four. The next problem is to quantify the entanglement of these graphs states in terms of multipartite entanglement measures, such as the geometric measure of entanglement and relative entropy of entanglement. Constructing the potential link between maximal entanglement and the mutually unbiased basis for dimension six may be a long-term goal of receiving more attentions.

\section*{Acknowledgments}
LC was supported by the Fundamental Research Funds for the Central Universities (Grant Nos. 30426401 and 30458601). DLZ was supported by NSF of China (Grants No.
11475254 and Grant No.  11175247) and NKBRSF of
China (Grant Nos. 2012CB922104 and 2014CB921202).

\appendix


\section{Proof of commutation relations}
\label{sec:proof-comm-relat}

In this appendix we give the proof of relations in Eqs. \eqref{eq:22}-\eqref{eq:33}, respectively.

The proof of Eq.~(\ref{eq:22}):

   Notice that the identity operator for the $m$-th particle is
  \begin{equation}
    \label{eq:23}
    I_{m} = \sum_{\alpha} \vert a_{\alpha}\rangle\langle
    a_{\alpha}\vert \equiv \vert a_{\alpha}\rangle\langle
    a_{\alpha}\vert,
  \end{equation}
  where we take the Einstein's rule for repeated indexes. Then
  \begin{eqnarray*}
    & & D_{m}(a_{i}) A_{m}(a_{j}) \\
    & = & D_{m}(a_{i}) A_{m}(a_{j}) I_{m} \\
    & = & D_{m}(a_{i}) A_{m}(a_{j}) \vert a_{\alpha}\rangle\langle
          a_{\alpha}\vert \\
    & = & D_{m}(a_{i}) \vert a_{\alpha} + a_{j}\rangle\langle
          a_{\alpha}\vert \\
    & = & \vert a_{i} a_{\alpha} + a_{i} a_{j}\rangle\langle
          a_{\alpha}\vert \\
    & = & A_{m}(a_{i} a_{j}) \vert a_{i} a_{\alpha}\rangle\langle
          a_{\alpha}\vert \\
    & = & A_{m}(a_{i} a_{j}) D_{m}(a_{i}).
  \end{eqnarray*}

  The proof of Eq.~(\ref{eq:24}):
  \begin{eqnarray*}
    & & C_{mn}\left(a_{i}\right)A_{m}\left(a_{j}\right)\\
    & = & C_{mn}\left(a_{i}\right)A_{m}\left(a_{j}\right)\vert
          a_{\alpha},a_{\beta}\rangle\langle a_{\alpha},a_{\beta}\vert\\
    & = & C_{mn}\left(a_{i}\right)\vert
          a_{\alpha}+a_{j},a_{\beta}\rangle\langle
          a_{\alpha},a_{\beta}\vert\\
    & = & \vert
          a_{\alpha}+a_{j},a_{\beta}+a_{i}a_{\alpha}+a_{i}a_{j}\rangle\langle
          a_{\alpha},a_{\beta}\vert\\
    & = & A_{m}\left(a_{j}\right)A_{n}\left(a_{i}a_{j}\right)\vert
          a_{\alpha},a_{\beta}+a_{i}a_{\alpha}\rangle\langle
          a_{\alpha},a_{\beta}\vert\\
    & = &
          A_{m}\left(a_{j}\right)A_{n}\left(a_{i}a_{j}\right)C_{mn}\left(a_{i}
          \right).
  \end{eqnarray*}

  The proof of Eq.~(\ref{eq:25}):
  \begin{eqnarray*}
    & & C_{mn}\left(a_{i}\right)A_{n}\left(a_{j}\right)\\
    & = & C_{mn}\left(a_{i}\right)A_{n}\left(a_{j}\right)\vert
          a_{\alpha},a_{\beta}\rangle\langle a_{\alpha},a_{\beta}\vert\\
    & = & C_{mn}\left(a_{i}\right)\vert
          a_{\alpha},a_{\beta}+a_{j}\rangle\langle
          a_{\alpha},a_{\beta}\vert\\
    & = & \vert a_{\alpha},a_{\beta}+a_{j}+a_{i}a_{\alpha}\rangle\langle
          a_{\alpha},a_{\beta}\vert\\
    & = & A_{n}\left(a_{j}\right)\vert
          a_{\alpha},a_{\beta}+a_{i}a_{\alpha}\rangle\langle
          a_{\alpha},a_{\beta}\vert\\
    & = & A_{n}\left(a_{j}\right)C_{mn}\left(a_{i}\right).
  \end{eqnarray*}

  The proof of Eq.~(\ref{eq:26}):
  \begin{eqnarray*}
    & & C_{mn}\left(a_{i}\right)D_{m}\left(a_{j}\right)\\
    & = & C_{mn}\left(a_{i}\right)D_{m}\left(a_{j}\right)\vert
          a_{\alpha},a_{\beta}\rangle\langle a_{\alpha},a_{\beta}\vert\\
    & = & C_{mn}\left(a_{i}\right)\vert
          a_{j}a_{\alpha},a_{\beta}\rangle\langle
          a_{\alpha},a_{\beta}\vert\\
    & = & \vert
          a_{j}a_{\alpha},a_{\beta}+a_{i}a_{j}a_{\alpha}\rangle\langle
          a_{\alpha},a_{\beta}\vert\\
    & = & D_{m}\left(a_{j}\right)\vert
          a_{\alpha},a_{\beta}+a_{i}a_{j}a_{\alpha}\rangle\langle
          a_{\alpha},a_{\beta}\vert\\
    & = & D_{m}\left(a_{j}\right)C_{mn}\left(a_{i}a_{j}\right).
  \end{eqnarray*}

  The proof Eq.~(\ref{eq:27}):
  \begin{eqnarray*}
    & & C_{mn}\left(a_{i}\right)D_{n}\left(a_{j}\right)\\
    & = & C_{mn}\left(a_{i}\right)D_{n}\left(a_{j}\right)\vert
          a_{\alpha},a_{\beta}\rangle\langle a_{\alpha},a_{\beta}\vert\\
    & = & C_{mn}\left(a_{i}\right)\vert
          a_{\alpha},a_{\beta}a_{j}\rangle\langle
          a_{\alpha},a_{\beta}\vert\\
    & = & \vert a_{\alpha},a_{\beta}a_{j}+a_{i}a_{\alpha}\rangle\langle
          a_{\alpha},a_{\beta}\vert\\
    & = & D_{n}\left(a_{j}\right)\vert
          a_{\alpha},a_{\beta}+a_{j}^{-1}a_{i}a_{\alpha}\rangle\langle
          a_{\alpha},a_{\beta}\vert\\
    & = & D_{n}\left(a_{j}\right)C_{mn}\left(a_{j}^{-1}a_{i}\right).
  \end{eqnarray*}

  If $A\neq0$, then

  \begin{eqnarray*}
    & & C_{mn}\left(a_{i}\right)C_{nm}\left(a_{j}\right)\\
    & = & C_{mn}\left(a_{i}\right)C_{nm}\left(a_{j}\right)\vert
          a_{\alpha},a_{\beta}\rangle\langle a_{\alpha},a_{\beta}\vert\\
    & = & C_{mn}\left(a_{i}\right)\vert
          a_{\alpha}+a_{j}a_{\beta},a_{\beta}\rangle\langle
          a_{\alpha},a_{\beta}\vert\\
    & = & \vert
          a_{\alpha}+a_{j}a_{\beta}, a_{\beta}+a_{i}a_{\alpha} +
          a_{i}a_{j}a_{\beta}\rangle\langle
          a_{\alpha},a_{\beta}\vert\\
    & = & \vert
          a_{\alpha}+a_{j}a_{\beta},A a_{\beta}+a_{i}a_{\alpha}\rangle\langle
          a_{\alpha},a_{\beta}\vert\\
    & = & D_{n}\left(A\right)\vert
          a_{\alpha}+a_{j}a_{\beta},a_{\beta}+\frac{a_{i}}{A}a_{\alpha}\rangle\langle
          a_{\alpha},a_{\beta}\vert\\
    & = &
          D_{n}\left(A\right)C_{nm}\left(a_{j}\right)\vert \frac{1}{A}
          a_{\alpha},a_{\beta} +
          \frac{a_{i}}{A}a_{\alpha}\rangle\langle
          a_{\alpha},a_{\beta}\vert\\
    & = &
          D_{n}\left(A\right)C_{nm}\left(a_{j}\right) D_{m}
          \left(\frac{1}{A}\right)\vert
          a_{\alpha},a_{\beta}+ \frac{a_{i}}{A}a_{\alpha}\rangle\langle
          a_{\alpha},a_{\beta}\vert\\
    & = &
          D_{n}\left(A\right)C_{nm}\left(a_{j}\right)D_{m}
          \left(\frac{1}{A}\right) C_{mn}\left(\frac{a_{i}}{A}\right)\\
    & = &
          D_{n}\left(A\right)D_{m}\left(\frac{1}{A}\right)C_{nm}
          \left(a_{j}A\right)C_{mn}\left(\frac{a_{i}}{A}\right).
  \end{eqnarray*}
  If $A=0$, then
  \begin{eqnarray*}
    & & C_{mn}\left(a_{i}\right)C_{nm}\left(a_{j}\right)\\
    & = & \vert a_{\alpha}+a_{j}a_{\beta},a_{i}a_{\alpha}\rangle\langle
          a_{\alpha},a_{\beta}\vert\\
    & = & W_{mn}\vert
          a_{i}a_{\alpha},a_{\alpha}+a_{j}a_{\beta}\rangle\langle
          a_{\alpha},a_{\beta}\vert\\
    & = & W_{mn}D_{m}\left(a_{i}\right)D_{n}\left(a_{j}\right)\vert
          a_{\alpha},a_{j}^{-1}a_{\alpha}+a_{\beta}\rangle\langle
          a_{\alpha},a_{\beta}\vert\\
    & = &
          W_{mn}D_{m}\left(a_{i}\right)D_{n}\left(a_{j}\right)
          C_{mn}\left(a_{j}^{-1}\right),
  \end{eqnarray*}
  where $W_{mn}$ is the swapp gate between the $m$-th qudit and the
  $n$-th qudit.

  The proof of Eq.~(\ref{eq:31}):
  \begin{eqnarray*}
    & & C_{mn}\left(a_{i}\right)C_{ml}\left(a_{j}\right)\\
    & = & C_{mn}\left(a_{i}\right)C_{ml}\left(a_{j}\right)\vert
          a_{\alpha},a_{\beta},a_{\gamma}\rangle\langle
          a_{\alpha},a_{\beta},a_{\gamma}\vert_{mnl}\\
    & = & C_{mn}\left(a_{i}\right)\vert
          a_{\alpha},a_{\beta},a_{\gamma}+a_{j}a_{\alpha}\rangle\langle
          a_{\alpha},a_{\beta},a_{\gamma}\vert_{mnl}\\
    & = & \vert
          a_{\alpha},a_{\beta}+a_{i}a_{\alpha},a_{\gamma}+a_{j}a_{\alpha}\rangle\langle
          a_{\alpha},a_{\beta},a_{\gamma}\vert_{mnl}\\
    & = & C_{ml}\left(a_{j}\right)\vert
          a_{\alpha},a_{\beta}+a_{i}a_{\alpha},a_{\gamma}\rangle\langle
          a_{\alpha},a_{\beta},a_{\gamma}\vert_{mnl}\\
    & = & C_{ml}\left(a_{j}\right)C_{mn}\left(a_{i}\right).
  \end{eqnarray*}

  The proof of Eq.~(\ref{eq:32}):
  \begin{eqnarray*}
    & & C_{mn}\left(a_{i}\right)C_{l n}\left(a_{j}\right)\\
    & = & C_{mn}\left(a_{i}\right)C_{l n}\left(a_{j}\right)\vert
          a_{\alpha},a_{\beta},a_{\gamma}\rangle\langle
          a_{\alpha},a_{\beta},a_{\gamma}\vert_{mnl}\\
    & = & C_{mn}\left(a_{i}\right)\vert
          a_{\alpha},a_{\beta}+a_{j}a_{\gamma},a_{\gamma}\rangle\langle
          a_{\alpha},a_{\beta},a_{\gamma}\vert_{mnl}\\
    & = & \vert
          a_{\alpha},a_{\beta}+a_{j}a_{\gamma}+a_{i}a_{\alpha},a_{\gamma}\rangle\langle
          a_{\alpha},a_{\beta},a_{\gamma}\vert_{mnl}\\
    & = & C_{l n}\left(a_{j}\right)\vert
          a_{\alpha},a_{\beta}+a_{i}a_{\alpha},a_{\gamma}\rangle\langle
          a_{\alpha},a_{\beta},a_{\gamma}\vert_{mnl}\\
    & = & C_{l n}\left(a_{j}\right)C_{mn}\left(a_{i}\right).
  \end{eqnarray*}

  The proof of Eq.~(\ref{eq:33}):
  \begin{eqnarray*}
    & & C_{mn}\left(a_{i}\right)C_{nl}\left(a_{j}\right)\\
    & = & C_{mn}\left(a_{i}\right)C_{nl}\left(a_{j}\right)\vert
          a_{\alpha},a_{\beta},a_{\gamma}\rangle\langle
          a_{\alpha},a_{\beta},a_{\gamma}\vert_{mnl}\\
    & = & C_{mn}\left(a_{i}\right)\vert
          a_{\alpha},a_{\beta},a_{\gamma}+a_{j}a_{\beta}\rangle\langle
          a_{\alpha},a_{\beta},a_{\gamma}\vert_{mnl}\\
    & = & \vert
          a_{\alpha},a_{\beta}+a_{i}a_{\alpha},a_{\gamma}+a_{j}a_{\beta}\rangle\langle
          a_{\alpha},a_{\beta},a_{\gamma}\vert_{mnl}\\
    & = & C_{nl}\left(a_{j}\right)\vert
          a_{\alpha}, a_{\beta}+a_{i}a_{\alpha}, a_{\gamma} -
          a_{i}a_{j}a_{\alpha}\rangle\langle
          a_{\alpha},a_{\beta},a_{\gamma}\vert_{mnl}\\  
    & = &
          C_{nl}\left(a_{j}\right)C_{mn}\left(a_{i}\right) C_{ml}
          \left(-a_{i}a_{j}\right).  
  \end{eqnarray*}

\bibliographystyle{unsrt}

\bibliography{graph}

\end{document}